\documentclass[a4paper,12pt,draft]{article}

\usepackage{amssymb,amsthm,amsmath,amsfonts}
\usepackage{epsfig}
\usepackage{tikz}

\theoremstyle{plain}
\begingroup

\newtheorem{thm}{Theorem}[section]
\newtheorem{theorem}{Theorem}[section]

\newtheorem{lemma}[thm]{Lemma}

\endgroup
\theoremstyle{definition}
\newtheorem{defn}{Definition}[section]
\newtheorem{definition}{Definition}[section]

\newtheorem{remark}[defn]{Remark}

\theoremstyle{remark}

\numberwithin{equation}{section}
\numberwithin{figure}{section}
\DeclareMathOperator{\re}{Re}

\def\I{\mathrm{i}}
\def\R{{\mathbb R}}
\def\C{{\mathbb C}}

\begin{document}

%%%%%%%%%%%%%%%%%%%%%%%%%%%%%%%%%%%

\title{Hadamard's variational formula and the energy-momentum tensor}

\author{Bj\"orn Gustafsson\textsuperscript{1}, Ahmed Sebbar\textsuperscript{2}\textsuperscript{3}}

\date{\today}

%\begin{document}

\maketitle

%\tableofcontents

%%%%%%%%%%%%%%%%%%%%%%%%%%%%%%%%%%

\begin{abstract} 
The Hadamard variational formula for the Green function is formulated in terms of a polarized energy-momentum tensor and a strain tensor.
This is elaborated in a general setting of subdomains of a Riemannian manifold in arbitrary dimension and linked to the way the energy-momentum
tensor in general field theory appears as a result of varying the metric tensor in a Lagrangian function.
\end{abstract}



 \footnotetext[1]
{Department of Mathematics, KTH, 10044, Stockholm, Sweden.\\
Email: \tt{gbjorn@kth.se}}
 \footnotetext[2]
{Department of Mathematics, Chapman University,  Orange, California 92866.\\
Email: \tt{sebbar@chapman.edu}}
\footnotetext[3]
{Universit\'{e} Bordeaux, IMB, UMR 5251, F-33405 Talence, France.\\
Email: \tt{ahmed.sebbar@math.u-bordeaux.fr}}

\noindent {\it Acknowledgements:} 
The first author expresses warm thanks for a generous invitation to Chapman University, where this work was initiated, 
and for the creative and friendly atmosphere which he experienced during the stay.

%%%%%%%%%%%%%%%%%%%%%%%%%%%%%%%%%%%%%%%%%%%%%%%%%%%%%%%%%%%%%%%%%%%%%%

\section{Introduction}\label{sec:introduction}

The Hadamard variational formula expresses how the Green function for a domain changes under an infinitesimal
variation of the boundary of the domain. It is usually formulated in terms of a boundary integral, like in (\ref{traditional}) below.
However, in his book \cite{Garabedian-1964}, Paul Garabedian formulated the principle instead in terms of an area integral (in two dimensions)
containing a generalization of the Maxwell stress tensor, which is an energy-momentum tensor for the electromagnetic field (see \cite{Landau-Lifshitz-1962}). 
The present paper grew out from attempts to understand Garabedian's point of view from a more general perspective.

We elaborate the subject in a general setting of subdomains of a Riemannian manifold of arbitrary dimension using tools of differential geometry
and tensor analysis. The Lie derivative plays a crucial role, even when re-deriving
the classical boundary integral formulation of Hadamard's principle in Section~\ref{traditional}.

The main result, Theorem~\ref{thm:Hadamard} in Section~\ref{sec:interior version}, expresses the Hadamard principle in terms of a volume integral
containing the energy-momentum tensor and a strain tensor. 
The Green function takes the role of representing, as a potential, the physical field in the energy-momentum tensor. 
This tensor is quadratic in the field, and in our case it is actually
polarized, with two different Green functions. The strain tensor contains the information of how the imposed vector field deforms the domain, and thereby
also deforms the metric tensor of the domain.

All this makes the treatment accord with general principles of physics, which are briefly discussed in the final Section~\ref{sec:physics}, partly in terms of an example
from \cite{Hawking-Ellis-1973}.

The present paper can be viewed as a continuation of our understanding of the Green function started in \cite{Gustafsson-Sebbar-2012}.

%%%%%%%%%%%%%%%%%%%%%%%%%%%%%%%%%%%%%%%%%%%%%%%%%%%%%%%%%%%%%%%%%%%%%%

\section{Traditional case, using the Lie derivative}\label{sec:traditional}

The (Laplacian) Green function $G_a$ for a (bounded) domain $\Omega\subset \R^n$ is defined by the properties
\begin{align*}
-\Delta G_a &=\delta_a\quad\text {in \,\,}\Omega, \\
G_a&=0 \quad\text {on \,\,}\partial \Omega.
\end{align*}
Writing $G(x,a)=G_a(x)$, $G(x,a)$ is symmetric with respect to $x$ and $a$. This is most clearly seen by using standard Green's formulas to express
the Green function as a mutual energy:
\begin{equation}\label{Gab0}
G(a,b)=\int_\Omega (\nabla G_a\cdot \nabla G_b)\,dx.
\end{equation}

The Green function certainly depends on the domain, $G=G_\Omega$, and Hadamard's classical formula \cite{Garabedian-1964}
expresses how $G_\Omega(x,a)$
changes under small (infinitesimal) variations of the boundary (there also exist formulas relating to various kinds of interior variations). 
Here we shall start by reviewing this formula (namely equation (\ref{traditional}) below) using the language of differential forms
and Lie derivatives. 

So let a smoothly bounded domain $\Omega= \Omega(t)\subset \R^n$ move in the flow of a vector field ${\bf v}=\sum_{j=1}^n v^j \frac{\partial}{\partial x^j}$, 
and denote by $\mathcal{L}_{\bf v}$ the Lie derivative, and by $i({\bf v})$ interior derivation (``contraction'' ), with respect to ${\bf v}$.
See in general Frankel \cite{Frankel-2012} for differential geometric concepts and notations.  
That $\Omega(t)$ moves in flow of ${\bf v}$ means effectively just that the boundary $\partial\Omega$ 
moves with speed, as measured in the normal direction, equal to the normal component of ${\bf v}$ on $\partial\Omega$.

One basic property of the Lie derivative is that if we integrate
a differential $p$-form $\omega$ over a $p$-chain $\gamma(t)$ which moves in the flow of a vector field $\bf v$ 
($t$ being the corresponding time parameter), then
\begin{equation}\label{Lomega}
\frac{d}{dt}\int_{\gamma(t)} \omega =\int_{\gamma(t)} \mathcal{L}_{\bf v}\omega.
\end{equation}
In case $\omega$ itself depends on $t$ there will be an additional term $\int_{\gamma(t)}\frac{\partial \omega}{\partial t}$. Also the vector field ${\bf v}$
may depend on $t$, but that causes no changes in the formula.  We shall need (\ref{Lomega}) only in the case $p=n$ and $\omega(t)=\Omega(t)$. 

Now, in the language of differential forms the representation (\ref{Gab0}) 
takes the form
\begin{equation}\label{Gab}
G(a,b)=\int_{\Omega} dG (\cdot, a)\wedge * dG(\cdot, b),
\end{equation}
where the star is the Hodge star.
From this we have, with $a,b\in \Omega$ $(a\ne b)$ kept fixed,
\begin{equation}\label{HadamardL}
\frac{d}{dt}G_{\Omega(t)} (a,b) =\int_{\Omega(t)} \mathcal{L}_{\bf v}(dG_\Omega (\cdot, a)\wedge * dG_\Omega(\cdot, b)),
\end{equation}
as an immediate consequence of (\ref{Lomega}).
The right member of (\ref{HadamardL}) can be made more explicit by the computation (suppressing the $t$ in $\Omega(t)$)
$$
\int_{\Omega} \mathcal{L}_{\bf v}(dG_a \wedge * dG_b)=\int_\Omega(d\circ i({\bf v})+i({\bf v})\circ d)(dG_a \wedge * dG_b)=
$$
$$
=\int_\Omega d(i({\bf v})(dG_a \wedge * dG_b))=\int_{\partial\Omega} i({\bf v})(dG_a \wedge * dG_b)=
$$
$$
=\int_{\partial\Omega} (i({\bf v})dG_a) \wedge * dG_b-\int_{\partial\Omega} dG_a \wedge i({\bf v})(*dG_b).
$$
Here the second term disappears since $dG_a=0$ along $\partial\Omega$. For the first term we have 
$$
i({\bf v})dG_a=\frac{\partial G_a}{\partial n} \,v_n \quad \text{on\,\,}\partial\Omega,
$$
$$
*dG_b=\frac{\partial G_b}{\partial n} \,d\sigma \quad \text{along\,\,}\partial\Omega,
$$
in terms of (outward) normal derivatives and components and with $d\sigma$ denoting the surface area element on $\partial \Omega$. Thus 
$$
\int_{\Omega} \mathcal{L}_{\bf v}(dG_a \wedge * dG_b)
=\int_{\partial\Omega}  \frac{\partial G_a}{\partial n} \frac{\partial G_b}{\partial n} \,v_n d\sigma,
$$
to be inserted in (\ref{HadamardL}).

In traditional notation, with
$$
\quad v_n=\frac{\delta n}{\delta t}, \quad \frac{d}{dt} G_{\Omega(t)}(a,b) =\frac{\delta}{\delta t} G(a,b),
$$
the variational formula takes the well-known form
\begin{equation}\label{traditional}
\delta G(a,b) =\int_{\partial\Omega}  \frac{\partial G(\cdot,a)}{\partial n} \frac{\partial G(\cdot,b)}{\partial n} \,\delta n \,d\sigma.
\end{equation}
The derivation remains valid in the Riemannian manifold setting to be discussed in the next section.

\begin{remark}
A case of special interest arises when the vector field ${\bf v}$ itself is generated by the Green function, with pole at a third point $c\in\Omega$,
namely when ${\bf v}=\nabla G(\cdot,c)$. In order for ${\bf v}$ to exist in a full neighborhood of $\partial{\Omega}$, that boundary 
need to be analytic. The formula (\ref{traditional}) becomes in this case
$$
\delta G(a,b) =\int_{\partial\Omega}  \frac{\partial G(\cdot,a)}{\partial n} \frac{\partial G(\cdot,b)}{\partial n}\frac{\partial G(\cdot,c)}{\partial n} \,\delta n \,d\sigma,
$$
with a triple symmetry. This beautiful formula appeared in \cite{Wiegmann-Zabrodin-2000} expressing a kind of integrability of the Dirichlet problem in relation
to Laplacian growth. See \cite{Gustafsson-Teodorescu-Vasiliev-2014} in this respect. 
\end{remark}

%%%%%%%%%%%%%%%%%%%%%%%%%%%
%%%%%%%%%%%%%%%%%%%%%%%%%%%

\section{Interior version by  energy-momentum tensor}\label{sec:interior version}

In the computation in Section~\ref{sec:traditional} we had an integral over $\Omega$ involving a Lie derivative, and this integral was pushed to the boundary.
But there is also the possibility not to go to the boundary. Then also the vector field ${\bf v}$ will be differentiated, and one may arrange matters so that
the derivatives of ${\bf v}$ appear only in a certain {\it strain tensor} $D$  (see below). 
The remaining factor will be a polarized {\it energy-momentum tensor} $T=T(a,b)$
(see in general \cite{Hawking-Ellis-1973, Frankel-2012}), or {\it Maxwell stress tensor} in the terminology of Garabedian \cite{Garabedian-1964}.

We shall work in the more general setting of an oriented Riemannian manifold $M$ with metric
$$
ds^2=g_{ij}(x) dx^i\otimes dx^j,
$$
and a relatively compact domain $\Omega\subset M$ with smooth boundary. Here and in the sequel the Einstein summation convention applies: any index 
which occurs once up and once down in a term is summed over, from $1$ to $n$. The volume form for the metric is 
$$
{\rm vol}^n= \sqrt{g}\,dx^1\wedge\dots\wedge dx^n,\quad \text{where \,}
g=\det (g_{ij}).
$$
The Green function $G_a(x)=G(x,a)=G_\Omega(x,a)$ of $\Omega$ is defined by
\begin{align*}
-d*d G_a &= \delta_a{\rm vol}^n \quad\text{in\,\,}\Omega,\\ G_a&=0  \quad\qquad\text{on\,\,} \partial \Omega.
\end{align*}
Here $\delta_a{\rm vol}^n$ is the unit point mass at $a\in\Omega$ considered as an $n$-form (an $n$-form current, more precisely).
The first equation can equivalently be written
$$
-\Delta G_a ={\delta}_a,
$$
where $\Delta$ is the Hodge Laplacian, here taking functions ($0$-forms) into functions. 

Let $D$ be the symmetric covariant tensor defined by
\begin{equation}\label{DLg}
2D=2D_{ij}(x)\,dx^i\otimes dx^j=\mathcal{L}_{\bf v}(g_{ij}dx^i\otimes dx^j).
\end{equation}
In elasticity theory $D$ is the strain tensor, measuring the deformation caused by ${\bf v}$. If there is no deformation, i.e. if $D=0$,
then ${\bf v}$ is called a Killing vector field. The components of $D$ are given by
$$
2D_{ij}=g_{ik}v^k_{;j}+g_{kj} v^k_{;i}=v_{i;j}+v_{j;i}\, 
$$
where semicolon refers to covariant differentiation, as is traditional. (Notational remark: Frankel \cite{Frankel-2012} uses a slash, in place of semicolon,
to denote covariant differentiation.) 

Next we define a symmetric tensor $T=T(a,b)=T(x;a,b)$, a polarized {energy-momentum} tensor with respect to the variable $x\in \Omega$,
depending on $a,b\in \Omega$ as parameters.  As a covariant tensor it is $T=T_{ij}dx^i\otimes dx^j$, where
$$
T_{ij}=\frac{\partial G(x,a)}{\partial x^i}\frac{\partial G(x,b)}{\partial x^j}+\frac{\partial G(x,a)}{\partial x^j}\frac{\partial G(x,b)}{\partial x^i}
-g_{ij}g^{k\ell}\,\frac{\partial G(x,a)}{\partial x^k}\frac{\partial G(x,b)}{\partial x^\ell}.
$$
This definition differs from what is common in physics by a factor two, but the above will be convenient for us.

To simplify notation we set
\begin{equation}\label{alphabetaPhi}
\alpha=dG(\cdot,a),\quad \beta =dG(\cdot,b), \quad \Phi =\alpha_k \beta^k.
\end{equation}
Then $\alpha$ and $\beta$ are $1$-forms, or covariant vector fields, 
while $\Phi$ is a scalar field which can be thought of as a mutual energy  density between $\alpha$ and $\beta$. In terms of components, 
$$
\alpha=\alpha_jdx^j, \quad \beta=\beta_j dx^j,
$$
where
\begin{equation}\label{Galpha}
\alpha_j=\frac{\partial G(x,a)}{\partial x^j}, \quad \beta_j=\frac{\partial G(x,b)}{\partial x^j}.
\end{equation}

We shall allow to freely raise and lower indices by means of the metric tensor. For example,
$$
\beta^j=\beta_i g^{ij},\quad g_{ij} g^{jk} = g_i^k=\delta_i^k \quad (\text{Kronecker delta}).
$$
Then we have 
\begin{equation}\label{TABn}
T_{ij} =\alpha_{i} \beta_{j} +\alpha_{j}\beta_{i}-g_{ij} \alpha_{k}\beta^{k}=\alpha_{i} \beta_{j} +\alpha_{j}\beta_{i}-\Phi\,g_{ij}.
\end{equation}
The trace of $T$ is
\begin{equation}\label{trace}
{\rm tr\,}T=T_{ij}g^{ij}= (2-n) \Phi.
\end{equation}
In the above notation, 
\begin{equation}\label{Gab1}
G(a,b)=\int_\Omega \alpha \wedge * \beta =\int_\Omega \Phi \,{\rm vol}^n.
\end{equation}

The contravariant version of the energy-momentum tensor has components
$$
T^{ij}=T_{rs}\,g^{ri}g^{s j}=\alpha^i \beta^j +\alpha^j \beta^i -\Phi\,g^{ij}.
$$
Using that all covariant derivatives of the metric tensor vanish this gives 
$$
T^{ij}_{;j}=\alpha^i_{;j} \beta^{j} +\alpha^i \beta^j_{;j} + \alpha^j_{;j} \beta^{i} +\alpha^{j} \beta^i_{;j}-\Phi_{;j}\,g^{ij}.
$$
By (\ref{Galpha}), 
$$
\alpha_{i;j}=\alpha_{j;i}, \quad \alpha_{i;j}=\alpha_{j;i}
$$
$$
\alpha^j_{;j} =-{\delta}_a, \quad \beta^j_{;j} =-{\delta}_b.
$$
It follows that $T^{ij}_{;j}=- \alpha^i {\delta}_b- \beta^i {\delta}_a$. Setting 
\begin{equation}\label{mu}
\mu^i =- \alpha^i {\delta}_b- \beta^i {\delta}_a=- \alpha^i (b){\delta}_b- \beta^i (a){\delta}_a
\end{equation}
for this source term we thus have:

\begin{lemma}\label{lem:divT}
The divergence of the energy-momentum tensor vanishes except for the two point source field $\mu=\mu (a,b)$ given by (\ref{mu}).
In terms of components,
\begin{equation}\label{divT}
T^{ij}_{;j}= \mu^i.
\end{equation}\end{lemma}

\begin{theorem}\label{thm:Hadamard}
The variation of the Green function $G_\Omega(a,b)$ due to a deformation of $\Omega\subset M$ driven by a smooth vector field ${\bf v}$
is, in terms of the energy-momentum tensor $T^{ij}=T^{ij}(a,b)$ and the strain tensor $D_{ij}$, given by
\begin{equation}\label{GTD}
\frac{d}{dt} G_{\Omega(t)} (a,b)=\int_{\Omega} T^{ij} D_{ij} \,{\rm vol}^n -v_i\mu^i =\int_\Omega(T^{ij}v_i)_{;j}\,{\rm vol}^n
\end{equation}
\end{theorem}

\begin{proof}
Using (\ref{Gab1}) we have
$$
\frac{d}{dt} G_{\Omega(t)} (a,b)=\int_\Omega\mathcal{L}_{\bf v}(\alpha\wedge *\beta) =\int_\Omega\mathcal{L}_{\bf v}(\Phi\, {\rm vol}^n)=
$$
$$
=\int_\Omega\mathcal{L}_{\bf v}(\Phi)\, {\rm vol}^n +\int_\Omega\Phi \,\mathcal{L}_{\bf v}({\rm vol}^n)=
$$
\begin{equation}\label{Phiv}
=\int_\Omega (\Phi_{,j} v^j +\Phi\, v_{;j}^j) \,{\rm vol}^n=\int_\Omega (\Phi \,v^j )_{;j} \,{\rm vol}^n. 
\end{equation}
Here we have used that $\mathcal{L}_{\bf v} (A\otimes B)=\mathcal{L}_{\bf v}(A)\otimes B+ A\otimes \mathcal{L}_{\bf v}(B)$ for arbitrary 
tensors $A$ and $B$ and that
$$
\mathcal{L}_{\bf v}({\rm vol})^n = d(i({\bf v}) {\rm vol}^n)=({\rm div}\,{\bf v})\, {\rm vol}^n= v_{;j}^j \,{\rm vol}^n
$$
(see \cite{Frankel-2012, Hawking-Ellis-1973} in general).
For an arbitrary vector field ${\bf A}=A^j\frac{\partial}{\partial x^j}$ we have, by Stokes' formula,
\begin{equation}\label{intA}
\int_\Omega A^j_{;j}{\rm vol}^n=\int_\Omega d( i({\bf A}) {\rm vol}^n)=\int_{\partial \Omega}i({\bf A}) {\rm vol}^n =\int_{\partial \Omega} A^j \,n_jd\sigma,
\end{equation}
where the oriented surface area form $n_jd\sigma$ is defined by the last equality. With $A^i=T^{ij}v_j$ this gives
$$
\int_{\partial\Omega} T^{ij}v_i\,n_jd\sigma= \int_\Omega (T^{ij}v_i)_{;j} {\rm vol}^n.
$$

Using the symmetries of $T$ and $D$ together with (\ref{mu}) and (\ref{divT}) the above identity can be continued as
$$
\int_\Omega(T^{ij}v_i)_{;j}\,{\rm vol}^n=\int_\Omega(T^{ij}_{;j}v_i+ T^{ij}v_{i;j})\,{\rm vol}^n
=-v_i\mu^i+\int_{\Omega} T^{ij} D_{ij} \,{\rm vol}^n.
$$
From this, and (\ref{Phiv}), we see that what remains to be proved is that
\begin{equation}\label{to prove}
\int_\Omega(T^{ij}v_i)_{;j}\,{\rm vol}^n=\int_\Omega (\alpha_k\beta^k\, v^j)_{;j}\,{\rm vol}^n.
\end{equation}

For the verification of (\ref{to prove}) we start with the left member and use (\ref{intA}) to turn some terms into boundary integrals:
$$
\int_{\Omega} (T^{ij} \,v_{i})_{;j} \,{\rm vol}^n=\int_\Omega\big((\alpha^i\beta^j+\alpha^j\beta^i -\alpha_k\beta^k g^{ij})\,v_{i}\big)_{;j} \,{\rm vol}^n=
$$
$$
=\int_\Omega\big((\alpha^i\beta^j+\alpha^j\beta^i -\alpha_k\beta^k g^{ij} -\alpha^k\beta_k g^{ij})\,v_{i}\big)_{;j} \,{\rm vol}^n+
$$
$$
+\int_\Omega (\alpha^k\beta_k\, g^{ij}v_i)_{;j}\,{\rm vol}^n=
$$
$$
=\int_\Omega\big((\alpha^i\beta^j-\alpha_k\beta^k g^{ij} )v_i+(\alpha^j\beta^i -\alpha^k\beta_k g^{ij})\,v_{i}\big)_{;j} \,{\rm vol}^n+
$$
$$
+\int_\Omega (\Phi\, v^j)_{;j}\,{\rm vol}^n=
$$
$$
=\int_{\partial\Omega}\big((\alpha^i\beta^j-\alpha_k\beta^k g^{ij} )v_i+(\alpha^j\beta^i -\alpha^k\beta_k g^{ij})\,v_{i}\big) \,n_jd\sigma+
$$
$$
+\int_\Omega (\Phi\, v^j)_{;j}\,{\rm vol}^n=
$$
\begin{equation}\label{boundary integral}
=\int_{\partial\Omega}\beta^j(\alpha_i n_j-\alpha_j n_i )\,v^id\sigma +\int_{\partial\Omega}\alpha^j(\beta_in_j -\beta_j n_i)\,v^{i}d\sigma+
\end{equation}
$$
+\int_\Omega (\Phi\, v^j)_{;j}\,{\rm vol}^n.
$$

Let $u$ be any defining function for $\partial\Omega$, i.e. any smooth function satisfying $u=0$ on $\partial\Omega$ and having nonzero gradient there.
In directions along the boundary $\partial\Omega$ we then have 
\begin{equation}\label{dGdGdu}
dG_a=0,\quad dG_b=0, \quad du=0, 
\end{equation}We can normalize $u$ so that  $|\nabla u|=1$ on $\partial\Omega$ and then $du=n_i dx^i$ where $n_i=\partial u/\partial x^i$ are the normal components in (\ref{boundary integral}).
Since $dG_a=\alpha_i dx^i$, $dG_b=\beta_i dx^i$ and since the covectors  $dG_a$, $dG_b$, $du$ by (\ref{dGdGdu}) are necessarily proportional at each point of $\partial\Omega$, it follows that
$$
\alpha_in_j=\alpha_jn_i, \quad \beta_i n_j=\beta_j n_i \quad \text{for all }\,\,i, j.
$$
Therefore the boundary integrals in (\ref{boundary integral}) disappear and we end up with
$$
\int_\Omega (T^{ij}v_i)_{;j}\,{\rm vol}^n=\int_\Omega (\Phi\, v^j)_{;j}\,{\rm vol}^n,
$$
which is (\ref{to prove}), as desired.
\end{proof}

%%%%%%%%%%%%%%%%%%%%%%
%%%%%%%%%%%%%%%%%%%%%%

\section{Connections to physics}\label{sec:physics}

In physics the energy $2$-form (here mutual energy) $\alpha\wedge *\beta$ is related to a Lagrangian density $L$ by 
$$
L\,{\rm vol}^n=\alpha\wedge *\beta,
$$
up to a sign and on disregarding source terms and mass terms. Thus $L=\Phi$, referring to (\ref{alphabetaPhi}). Theorem~\ref{thm:Hadamard} mainly expresses, in this language, that
\begin{equation}\label{LLTD}
\int_\Omega\mathcal{L}_{\bf v}(L\,{\rm vol}^n)=\int_{\Omega} T^{ij} \,D_{ij} \,{\rm vol}^n.
\end{equation}
This is in line with the general philosophy in physics that the energy-momen\-tum tensor arises from the Lagrangian as a result of an infinitesimal variation of the metric tensor.
See \cite{Hawking-Ellis-1973, Deligne-Freed-1999} for example. Recall also (\ref{DLg}) in this respect. 

In our case we started in Sections~\ref{sec:traditional} and \ref{sec:interior version} with a fixed manifold $M$ and let a subdomain $\Omega$ deform by moving in the flow of
a vector field ${\bf v}$. This means that $\Omega$, as a body, consists of the same material particles all the time. In that sense $\Omega$ is a fixed space (or body),
and what really changes is the metric tensor and the coordinate values for the individual particles. Thus we are within the realm of the above mentioned philosophy.
In physics books the formula (\ref{LLTD}) may look like
$$
\delta S=\int T^{ij }\, \delta g_{ij},
$$
where
$$
S=\int L\,{\rm vol}^n
$$ 
is the ``action''.

To set our discussions in a specific physical context we give an example of a Lagrangian function for scalar field $\psi$.
Citing from Hawking-Ellis \cite{Hawking-Ellis-1973} (Example~1 in Section~3.3  there), the Lagrangian for a field $\psi$ representing a $\pi^0$-meson is
$$
L=-\psi_{;i}\psi_{;j}g^{ij}-\frac{m^2}{\hbar^2}\psi^2.
$$
The Euler-Lagrange equation obtained by variation of the action is
$$
\psi_{;ij}g^{ij}-\frac{m^2}{\hbar^2}\psi^2=0.
$$

In our case we have a polarized Lagrangian with two different fields, 
and we also have a source term instead of a mass term. Our Lagrangian density would be 
(with a different sign convention compared to \cite{Hawking-Ellis-1973})
$$
L=(\psi_a)_{;i}(\psi_b)_{;j}g^{ij} -\psi_a\delta_b-\psi_b\delta_a,
$$
where $\psi_a$ and $\psi_b$ are independent scalar fields vanishing on $\partial\Omega$. Variation of the action 
gives the Euler Lagrange equations
$$
((\psi_a)_{;i}g^{ij})_{;j}=-\delta_a, \quad ((\psi_b)_{;i}g^{ij})_{;j}=-\delta_b,
$$ 
which are the defining functions for the Green functions $G_a$ and $G_b$.
Therefore the Lagrangian density eventually (``on-shell'') comes out as
$$
L=(G_a)_{;i}(G_b)_{;j}g^{ij}- G_a\delta_b-G_b\delta_a.
$$
Rewriting and including the volume form we have
$$
L\,{\rm vol}^n=\alpha\wedge *\beta-G(a,b)(\delta_a+\delta_b){\rm vol}^n.
$$

The counterpart of the energy-momentum tensor in \cite{Hawking-Ellis-1973}, namely
$$
T_{ij}=\psi_{;i}\psi_{;j}-\frac{1}{2}g_{ij}\big( \psi_{;k}\psi_{;\ell}g^{k\ell}+\frac{m^2}{\hbar^2}\psi^2    \big),
$$
is in our case (up to a factor two, and replacing immediately the fields $\psi_a$, $\psi_b$ by $G_a$, $G_b$)
$$
T_{ij}=(G_a)_{;i} (G_b)_{;j} +(G_a)_{;j} (G_b)_{;i}-g_{ij}g^{k\ell} (G_a)_{;k} (G_b)_{;\ell}.
$$
This is exactly what we started with in Section~\ref{sec:interior version}.

%%%%%%%%%%%%%%%%%%%%%%%%%%%%%%%%%%%
%%%%%%%%%%%%%%%%%%%%%%%%%%%%%%%%%%%

%%%%%%%%%%%%%%%%%%%%%%%%%%%%%%%%%%
%%%%%%%%%%%%%%%%%%%%%%%%%%%%%%%%%%

\bibliography{bibliography_gbjorn.bib}

\end{document}